%% file: treelapewplat.tex
\documentclass[preprint,12pt,times,numbers]{elsarticle}

\include{macros}
\include{shortcuts}

\usepackage{natbib}
\usepackage{version}
\usepackage{algorithm,algorithmic}
\excludeversion{hide}

\begin{document}

\begin{frontmatter}

\title{Tree simplification and the `plateaux' phenomenon of graph Laplacian eigenvalues}
\author[ucdmath]{Naoki Saito\corref{saitocor}}
\ead{saito@math.ucdavis.edu}
\cortext[saitocor]{Corresponding author}
\author[ucdmath]{Ernest Woei\fnref{woei}}
\ead{woei@math.ucdavis.edu}
\fntext[woei]{Currently at FlashFoto, Inc., Los Gatos, CA 95030, USA.}
\address[ucdmath]{Department of Mathematics, University of California, Davis, CA 95616, USA}

\input{abstract}

\begin{keyword}
vertex reduction; graph Laplacian eigenvalues; 
eigenvalue multiplicity; monic polynomials with integer coefficients

\MSC 05C07 \sep 05C50 \sep 15A42 \sep 65F15
\end{keyword}

\end{frontmatter}

\input{tree-simp}


\end{document}

%% file: macros.tex
\usepackage{amsmath,amssymb,amsthm,amsfonts,array,latexsym,eucal,mathrsfs}
\usepackage{fourier}
\usepackage{xparse}
\usepackage{cite}
\usepackage{color}
\usepackage{graphicx} 
\usepackage{subfigure}
\usepackage{caption}
\usepackage{tikz}
\usepackage{kbordermatrix}
\usepackage{placeins}

\usepackage{verbatim} 

\usepackage{todonotes} 
\usepackage{footnote}

\usepackage{ifthen}
\newtheorem{Theorem}{Theorem}[section]
\newtheorem{Proposition}[Theorem]{Proposition}
\newtheorem{Lemma}[Theorem]{Lemma}

\theoremstyle{definition}

\newtheorem{Example}[Theorem]{Example}
\newtheorem{Conjecture}[Theorem]{Conjecture}

\theoremstyle{remark}

\newcommand{\newchapter}[3] 
{                                                        
  \chapter[#2]{#3}
  \chaptermark{#1}
  \thispagestyle{myheadings}
}

\definecolor{myblue}{rgb}{0,0,.75}				
\definecolor{mylightblue}{rgb}{0,0,1}				
\definecolor{mytealblue}{rgb}{0.25,1,0.75}		
\definecolor{myspringgreen}{rgb}{0.75,1,0.25} 		
\definecolor{myred}{rgb}{1,0,0} 				
\definecolor{mysepia}{rgb}{0.5,0,0}			 	

%% file: shortcuts.tex
\newcommand{\Lvert}[1]{\left| #1 \right|}
\newcommand{\Lpar}[1]{\left( #1 \right)}
\newcommand{\Lbra}[1]{\left\{ #1 \right\}}

\newcommand{\eqn}[1]{\begin{equation} #1 \end{equation}}

\newcommand{\bphi}{\boldsymbol \phi}

\newcommand{\blambda}{\boldsymbol \lambda}

\NewDocumentCommand\ev{mo}{
  \IfNoValueTF{#2}
    {
      \IfNoValueTF{#1}
        {{\bphi}}
        {{\bphi_{#1}}}
    }
    {
      \IfNoValueTF{#1}
        {{\bphi(#2)}}
        {{\bphi_{#1}(#2)}}
    }
}

\NewDocumentCommand\bew{mo}{
  \IfNoValueTF{#2}
    {
      \IfNoValueTF{#1}
        {{\blambda}}
        {{\blambda_{#1}}}
    }
    {
      \IfNoValueTF{#1}
        {{\blambda(#2)}}
        {{\blambda_{#1}(#2)}}
    }
}

\newcommand{\bx}{\boldsymbol x}
\newcommand{\by}{\boldsymbol y}

\newcommand{\br}{\boldsymbol r}
\newcommand{\bxew}{\bx_{_\ew}}

\newcommand{\Rf}{{\mathbb R}}

\newcommand{\Q}{{\mathbb Q}}

\newcommand{\define}{{\, := \,}}
\newcommand{\cond}{\,|\,}

\newcommand{\vo}{V_1}
\newcommand{\vto}{V_{2 \sim 1}}
\newcommand{\vi}{V_I}
\newcommand{\voi}{\vo^I}
\newcommand{\vtoi}{\vto^I}

\newcommand{\tvi}{\tau_{_{\vi}}}

\newcommand{\dv}[1]{d_{v_{#1}}}

\newcommand{\lm}{\ew_-}
\newcommand{\lp}{\ew_+}
\newcommand{\ew}{{\lambda}}

\newcommand{\zv}{{\bf 0}}

\newcommand{\zvcv}[1]{\zv_{2c(v_{#1})}}

\DeclareMathOperator*{\diag}{diag}

\newcommand{\transp}{{\scriptscriptstyle{\mathsf{T}}}}

%% file: abstract.tex
\begin{abstract}
We developed a procedure of reducing the number of vertices and edges of a given
tree, which we call the ``tree simplification procedure,'' without changing
its topological information.
Our motivation for developing this procedure was to reduce computational
costs of graph Laplacian eigenvalues of such trees.
When we applied this procedure to a set of trees representing dendritic 
structures of retinal ganglion cells of a mouse and computed their graph
Laplacian eigenvalues, we observed two ``plateaux'' (i.e., two sets of multiple
eigenvalues) in the eigenvalue distribution of each such simplified tree.
In this article, after describing our tree simplification procedure, 
we analyze why such eigenvalue plateaux occur in a simplified tree, and 
explain such plateaux can occur in a more general graph if it satisfies a 
certain condition, identify these two eigenvalues specifically as well as
the lower bound to their multiplicity.
\end{abstract}

%% file: tree-simp.tex
\section{Introduction}
\label{sec:intro}

In order to characterize and cluster dendritic trees of retinal ganglion cells 
(RGCs) of a mouse, our previous work \citep{SAITO-ALL-2009} illustrated the use 
of graph Laplacian eigenvalues rather than using the morphological features 
derived manually from those dendritic trees.
We note that each dendritic tree was literarily represented by a tree
in the sense of graph theory.
Furthermore, in \citep{SAITO-ALL-2011, NAKATSUKASA-ALL-2013}, we provided our
theoretical understanding of the peculiar eigenvalue/eigenvector phase 
transition phenomenon we observed on each of our dendritic trees.

Continuing on our road to characterizing dendritic trees, once more we observed
a new eigenvalue phenomenon on our simplified (or vertex-subsampled) dendritic
trees.
Discovery of this phenomenon has the following history.
Each of the dendritic trees first analyzed in \citep{SAITO-ALL-2009} has
a large number of vertices (ranging from 565 to 24474 depending on the RGCs)
since they represent dense spatial sample points traced along the actual 
dendritic arbors in the 3D images measured by a confocal
microscope using specialized segmentation software operated by our 
neuroscience collaborators; see \citep{COOMBS-ALL-2006} for the details.  
If we are only concerned about their \emph{topological} properties (e.g.,
connectivities of various branches, bifurcation patterns, etc.) but not
their \emph{geometric aspects} (e.g., branch lengths, branch angles, etc.), 
then it is unnecessary to keep most of the vertices of degree 2 since
they do not alter topology of the trees.
Section~\ref{sec:tree-simp} describes our procedure to eliminate such
vertices, which we call a ``tree simplification procedure.''
After simplifying all these trees, we once again started analyzing their
graph Laplacian eigenvalues.  Interestingly enough, the phase transition
phenomenon we observed in the original trees disappeared completely.
Instead, we observed a new phenomenon in the eigenvalue distributions of
such simplified trees.  For each such simplified tree, the eigenvalue
distribution has two ``plateaux,'' i.e., a pair of distinct eigenvalues 
having the same multiplicity greater than one as shown in
Figures~\ref{fig:rgc102simpewplot} and~\ref{fig:rgc108simpewplot}.
The purpose of this article is to describe this phenomenon and provide a 
theoretical explanation of this phenomenon.

The organization of this article is the following.  
Section~\ref{sec:nottools} sets up our notation for this article and
defines some basic quantities.
Then, Section~\ref{sec:tree-simp} describes our tree simplification procedure 
and illustrates several examples.
Section~\ref{sec:ewmult} describes our main result along with its theoretical
consequences.
Finally, we conclude in Section~\ref{sec:conc} with discussion and state a
conjecture on more general situations.

\section{Notation and Definitions}
\label{sec:nottools}
In this article, we use the standard graph theory notation.
We mostly follow the commonly used notation; see e.g., \citep{GODSIL-ROYLE}.
Let $G=(V,E)$ be a graph where $V = V(G) = \{ v_1, v_2, \ldots, v_n \}$ is a 
\emph{vertex set} of $G$ and $E = E(G) = \{ e_1, e_2, \ldots, e_m \}$ is 
its \emph{edge set} where $e_k$ connects two vertices $v_i, v_j$ for some 
$1 \leq i \neq j \leq n$.  $\Lvert{V} = n$ is also referred to as the 
\emph{order} of $V$, where $\Lvert{\cdot}$ denotes a cardinality of a set.
We only deal with finite $n$ and $m$ in this article.
A \emph{subgraph} $H$ of a graph $G$ is a graph such that $V(H) \subseteq V(G)$
and $E(H) \subseteq E(G)$.
An edge connecting a vertex $v_i \in V$ and itself is called a \emph{loop}.
If there exist more than one edge connecting some $v_i, v_j \in V$,
then they are called \emph{multiple edges}.
A graph having loops or multiple edges is called a \emph{multiple graph}
(or \emph{multigraph}); otherwise it is called a \emph{simple} graph.

An edge $e \in E$ connecting two distinct vertices $v_i, v_j \in V$ may 
or may not have a direction.
If $e$ is an \emph{undirected} edge between $v_i$ and $v_j$, 
then we write $e=v_iv_j$, and $v_i, v_j$ are called the \emph{endpoints} of 
$e$. We also say an edge $e=v_iv_j$ is \emph{incident with} $v_i$ and $v_j$,
and $e$ \emph{joins} $v_i$ and $v_j$.
If $e=v_iv_j$, then $v_i, v_j$ are said to be \emph{adjacent} and we write 
$v_i \sim v_j$.
Let $N(u)$ be a set of \emph{neighbors} of vertex $u$, i.e., 
$N(u) \define \Lbra{v \in V \vert u \sim v}$.
An \emph{undirected graph} is a graph where none of its edges has a direction. 

If each edge $e \in E$ has a \emph{weight} (normally positive),
written as $w_e$, then $G$ is called a \emph{weighted} graph. 
$G$ is said to be \emph{unweighted} if $w_e \equiv 1$ for each $e \in E$.

A \emph{path} from $v_i$ to $v_j$ in a graph $G$ is a subgraph of $G$ 
consisting of a sequence of distinct vertices starting with $v_i$ and ending 
with $v_j$ such that consecutive vertices are adjacent.  
We say the \emph{length} (or \emph{cost}) $\ell(P)$ of a path $P$ is the
sum of its corresponding edge weights, i.e.,
$\ell(P) \define \sum_{e \in E(P)} w_e$.  
For any two vertices in $V$, if there is a path connecting them,
then such a graph is said to be \emph{connected}.
In this article, we only deal with \emph{simple, undirected, connected, and 
unweighted} graphs.

A path starting from $v_i$ that returns to $v_i$ (but is not a loop)
is called a \emph{cycle}.
A \emph{tree} is a connected graph without cycles, and
is often denoted by $T$ instead of $G$.  For a tree $T$,
we have $| E(T) | = | V(T) | -1$.
A \emph{rooted tree} $T(V,E)$ is a tree with a vertex labeled as the \emph{root} 
vertex, $v_r \in V$.  A graph $G'(V',E')$ \emph{spans} $G(V,E)$, if $G'$ is a 
tree with $V'=V$ and $E' \subseteq E$.  
Such $G'(V',E')$ is also called a \emph{spanning tree} of $G(V,E)$.

The \emph{Laplacian matrix} of $G$ is $L(G) \define D(G)-A(G)$ where 
$D(G) \define \diag\Lpar{d_{v_1},\dots,d_{v_n}}$ and $A(G)=\Lpar{a_{ij}}$
are the \emph{degree matrix} and the \emph{adjacency matrix} of $G$, 
respectively.  The entries of the latter are defined as
$$
a_{ij}= \begin{cases}
1 & \text{if $v_i \sim v_j$}; \\
0 & \text{otherwise}.
\end{cases}
$$
The \emph{degree} of vertex $v_i$ of $G$ is defined as
$d(v_i) = d_{v_i} \define \sum_{j=1}^n a_{ij}$, i.e., the $i$th row sum of $A(G)$.
A vertex $v \in V(G)$ is referred to as a \emph{leaf} or a \emph{pendant vertex}
if $d_v = 1$.

Let $0 = \ew_0 \leq \ew_1 \leq \cdots \leq \ew_{n-1}$ be the sorted eigenvalues
of $L(G)$.  Let $m_G(\ew)$ denote the multiplicity of the eigenvalue $\ew$ of 
$L(G)$, and let $m_G(I)$ be the number of eigenvalues of $L(G)$, multiplicities 
included, that belong to $I$, an interval of the real line.

In this article, we refer to a vertex of degree 2 in a given rooted tree that is
not a root vertex as a \emph{trivial} vertex.  Hence, a \emph{nontrivial} vertex 
means that its degree is other than 2 or it is a root vertex.  
For a graph $G(V,E)$, a pair of nontrivial vertices, say, $(v_{i_1}, v_{i_k}) \in V \times V$, is referred to as a \emph{neighboring pair of nontrivial vertices}
if there is a path between them with a vertex sequence 
$\Lpar{v_{i_1}, v_{i_2}, \dots, v_{i_k}}$
where the intermediate vertices $v_{i_2}, \ldots, v_{i_{k-1}}$ are all trivial.

A \emph{starlike} tree is a tree which has exactly one vertex of degree greater
than 2.  Let $S\Lpar{n_1,n_2,\dots,n_k}$ be a starlike tree that has $k(\geq 3)$
paths (i.e., branches) emanating from the central vertex $v_1$ with $d_{v_1}=k$.  
Let the $i$th branch have $n_i$ vertices excluding $v_1$.  
Let $n_1 \geq n_2 \geq \cdots \geq n_k$. Then, $n=\Lvert{V(S(n_1, \ldots, n_k))}=
n=1+\sum_{i=1}^k n_i$.

A \emph{bipartite graph} $G(V,E)$ is a graph such that the vertex set $V$ can be
partitioned into two disjoint sets, $U'$ and $V'$, i.e., $U' \cup V'=V$ and 
$U' \cap V' = \emptyset$, and for every pair of vertices in each of the disjoint
sets there is no edge that connects the pair of vertices, i.e., 
for all $u, u' \in U'$, we have $u \not\sim u'$ and $uu' \notin E$, 
similarly for $V'$.  Additionally, for each $u \in U'$, there exist a $v \in V'$ 
such that $uv \in E$.

\section{Tree Simplification}
\label{sec:tree-simp}
As mentioned earlier, we only deal with \emph{simple, connected, undirected, and 
unweighted} graphs in this article as we did in our previous works 
\citep{SAITO-ALL-2009, SAITO-ALL-2011, NAKATSUKASA-ALL-2013}.  In other words, 
we focus on the topological aspects of graphs rather than geometrical aspects.  
For the latter, we refer the readers to \citep[Chap.~7]{WOEI-2012} as
well as our ongoing work \citep{SAITO-ALL-2013-CLUSTTREES}.
We also assume that each tree we deal with in this article is a rooted tree
whose root vertex corresponds to the so-called ``soma'' (a.k.a.\ cell body)
if that tree represents an actual neuronal dendritic tree.
Once we decide to restrict our attention to the topological aspects of trees, 
then it seems obvious that all trivial vertices defined earlier can be removed 
without altering the topology and connectivities.  
Removal of such trivial vertices certainly saves subsequent computations
of eigenvalues and eigenvectors of the resulting Laplacian matrices.

Can we really remove all trivial vertices in an original tree?  
If we dealt with topological and geometrical aspects of trees by using weighted
trees with edge weight representing the Euclidean distance between the associated
pair of vertices, then the answer would be `Yes'. 
However, the answer is in fact `No' for unweighted trees of our interest.
This is due to the existence of the so-called \emph{spines}\footnote{
Spines are small membranous protrusion along a neuron's dendrite. 
They typically receive input from a single synapse of another neuron's axon. 
They also serve as a storage site for synaptic strength and help transmit 
electrical signals to the neuron's cell body \citep{STUART-ALL-2008}. 
Therefore, spines are a very important feature of dendrites.}
in our dendritic trees.  Each spine is represented by a \emph{pendant edge}
(i.e., an edge connecting a leaf and one of the intermediate vertices in a path);
see, e.g., Figure~\ref{fig:rgc60orivssimp}.

We want to distinguish spines from longer paths in our resulting simplified 
trees.  To do so, we need to do the following:
for each pair of \emph{neighboring} nontrivial vertices in an original tree, 
check the length of its associated path; if it is greater than one (i.e., 
a non-spine path), then we remove all those intermediate trivial vertices but 
one in the middle of the path, which results in the path of length two.
Note that this tree simplification procedure keeps all the spines intact.

Let us briefly illustrate the benefits we obtain from this tree simplification 
procedure on our dataset by displaying a histogram of an agglomeration of 
all the vertices and their degrees over all dendritic trees in 
Figure~\ref{subfig:orirgcaggdeghist}.  There we can see that the number of
trivial vertices outnumber the number of nontrivial vertices.  
Most trivial vertices in our original dendritic trees lie on paths between 
neighboring nontrivial vertices. 
Figure~\ref{subfig:od2noprgcaggdeghist} shows the degree distribution
of the simplified dendritic trees after our tree simplification procedure
is applied.  We can see that the number of degree 2 vertices in the simplified
trees are comparable with the nontrivial vertices with degree 1 and degree 3.
See also Table~\ref{tab:orisimpstats} for more quantitative information on
our tree simplification procedure applied to our dendritic trees.
\begin{figure}
\centering
\subfigure[]{
\includegraphics[width=.45\textwidth]{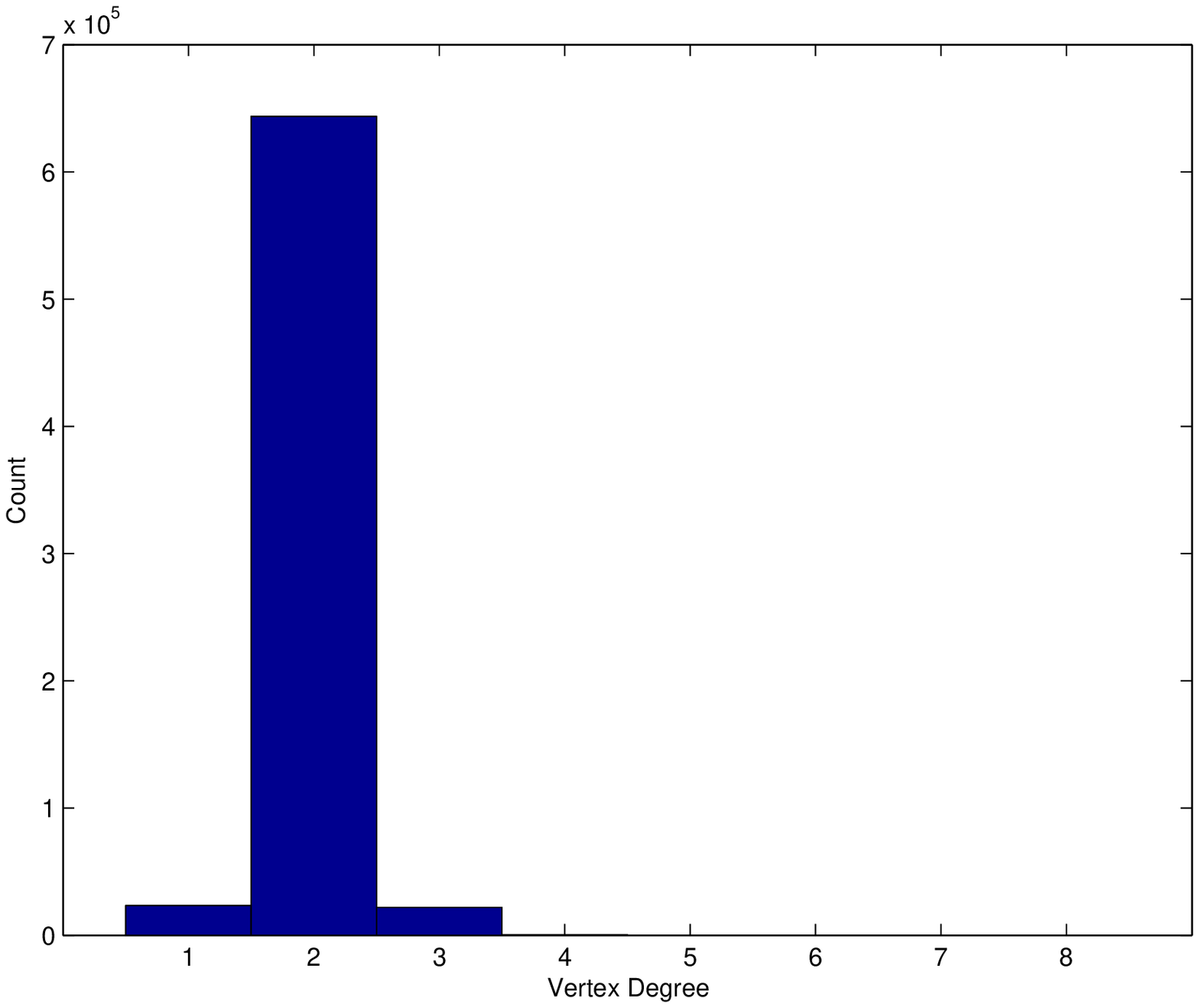}
\label{subfig:orirgcaggdeghist}
}
\subfigure[]{
\includegraphics[width=.45\textwidth]{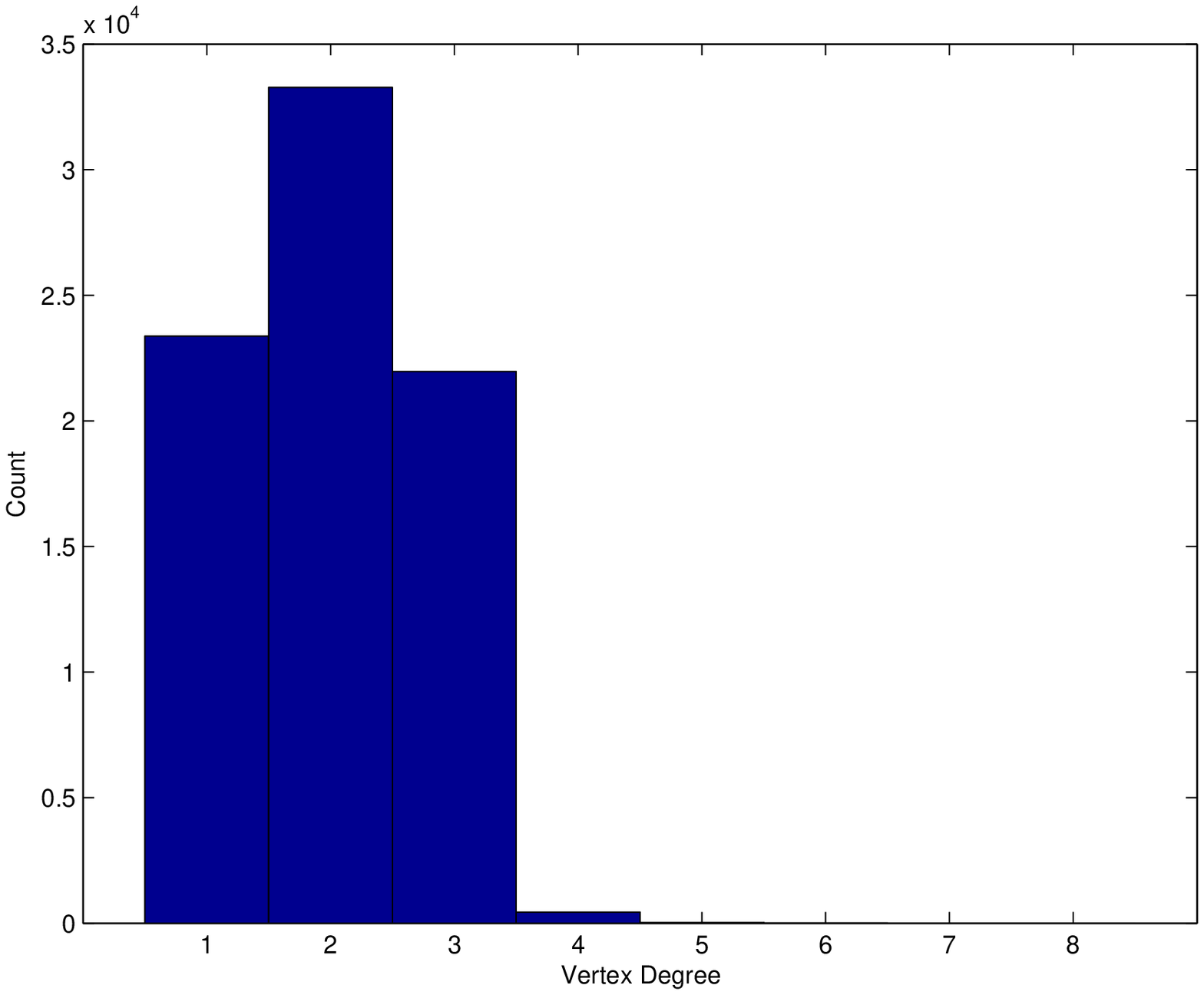}
\label{subfig:od2noprgcaggdeghist}
}
\caption{Histogram of an agglomeration of all vertices and their degrees for (a) unmodified dendritic trees; (b) simplified dendritic trees with one trivial vertex on each non-spine path. Note that the scale of the vertical axis in (a) is 
different from that in (b).}
\label{fig:aggdeghist}
\end{figure}

Let us now describe our tree simplification procedure in detail.
Let $T(V,E)$ be a rooted tree.  
Assume $V$ is an ordered list of vertices that are labeled as 
$v_i$ for $i=1,\dots,n=|V|$. 
Without loss of generality, let us assume $v_1$ is the labeled root vertex.
Our procedure starts by examining vertex $v_1$, then $v_2$, and so on till we
examine the last vertex $v_n$.
\begin{hide}
\begin{description}
\item[Step 0:] Set $i=0$.
\item[Step 1:] If $i > n$, then terminate; else set $i=i+1$ and $v=v_i$.
\item[Step 2:] If $v$ is $v_1$ or $v \notin V$, then goto \textbf{Step 1}.
\item[Step 3:] If $d_v \neq 2$, then goto \textbf{Step 1}.
\item[Step 4:] If $d_v=2$, then $v$ is adjacent to two vertices.
Let $u$ and $w$ be these two vertices.  
If $u$ or $w$ is $v_1$, then goto \textbf{Step 1}.
\item[Step 5:] If $d_u = 2$, then `coalesce' $u$ and $v$
(i.e., delete the edges $(v,u)$, $(v,w)$ from $E$;
add a new edge $(u,w)$ to $E$; and delete $v$ from $V$),
and goto \textbf{Step 1}.
\item[Step 6:] If $d_w = 2$, then follow \textbf{Step 5} with $w$ instead of $u$.
\item[Step 7:] Go to \textbf{Step 1}.
\end{description}
\end{hide}
\begin{algorithm}[H]
\caption{Tree Simplification}
\label{alg:tree-simp}
\begin{algorithmic}[1]
\STATE Set $i=0$.

\STATE If $i > n$, then terminate; else set $i=i+1$ and $v=v_i$.

\STATE If $v$ is $v_1$ or $v \notin V$, then goto \textbf{Step 2}.

\STATE If $d_v \neq 2$, then goto \textbf{Step 2}.

\STATE If $d_v=2$, then $v$ is adjacent to two vertices.
Let $u$ and $w$ be these two vertices.  
If $u$ or $w$ is $v_1$, then goto \textbf{Step 2}.

\STATE If $d_u = 2$, then `coalesce' $u$ and $v$
(i.e., delete the edges $(v,u)$, $(v,w)$ from $E$;
add a new edge $(u,w)$ to $E$; and delete $v$ from $V$),
and goto \textbf{Step 2}.

\STATE If $d_w = 2$, then follow \textbf{Step 6} with $w$ instead of $u$.

\STATE Go to \textbf{Step 2}.

\end{algorithmic}
\end{algorithm}
A few remarks are in order.
First, note that from \textbf{Step 3} to \textbf{Step 7}, we would coalesce 
\emph{at most} one pair of vertices.
%
Second, upon the completion of the above procedure, the pair of sets $(V,E)$
form the desired simplified tree.
The number of nontrivial vertices in the simplified tree is preserved by
this simplification procedure as can be seen from \textbf{Step 3} and 
\textbf{Step 4}.

Now let us illustrate our tree simplification procedure by applying it to 
very simple two trees in Figure~\ref{fig:p6top4} and \ref{fig:p5top5}. 
Figure~\ref{fig:p6top4} shows a simplification of tree $P_6$ with root vertex 
$v_r$ to $P_4$.  
In Figure~\ref{fig:p5top5}, the tree simplification procedure keeps $P_5$ 
intact since $v_r$ is adjacent to two trivial
vertices, hence we ``go back" to \textbf{Step 2} in \textbf{Steps 3} and 
\textbf{5}.

Figures~\ref{fig:rgc100orivssimp} and~\ref{fig:rgc60orivssimp} illustrate the
simplified trees of a few dendritic trees in our dataset.  
Table~\ref{tab:orisimpstats} provides information about how
much we simplified (i.e., subsampled) our original dendritic trees.
From this table, we see that on average, approximately 84\% of vertices
of our dendritic trees are removed by our procedure.
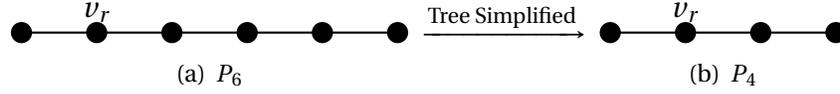
\begin{figure}
\centering
\subfigure[$P_6$]{
\begin{tikzpicture}[style=thick]
  \draw (0,0) -- (5,0);
  \foreach \x in {0,1,...,5} { \fill[black] (\x,0) circle (4pt); }
  \draw (1,0) node [above] {$v_r$};
\end{tikzpicture}
}
$\xrightarrow{\text{Tree Simplified}}$
\subfigure[$P_4$]{
\begin{tikzpicture}[style=thick]
  \draw (0,0) -- (3,0);
  \foreach \x in {0,1,...,3} { \fill[black] (\x,0) circle (4pt); }
  \draw (1,0) node [above] {$v_r$};
\end{tikzpicture}
}
\caption{$P_6$ simplified to $P_4$.}
\label{fig:p6top4}
\end{figure}%
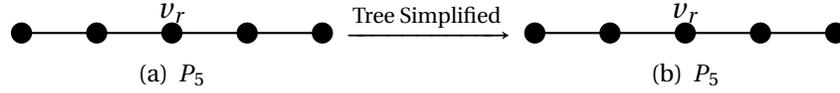
\begin{figure}
\centering
\subfigure[$P_5$]{
\begin{tikzpicture}[style=thick]
  \draw (0,0) -- (4,0);
  \foreach \x in {0,1,...,4} { \fill[black] (\x,0) circle (4pt); }
  \draw (2,0) node [above] {$v_r$};
\end{tikzpicture}
}
$\xrightarrow{\text{Tree Simplified}}$
\subfigure[$P_5$]{
\begin{tikzpicture}[style=thick]
  \draw (0,0) -- (4,0);
  \foreach \x in {0,1,...,4} { \fill[black] (\x,0) circle (4pt); }
  \draw (2,0) node [above] {$v_r$};
\end{tikzpicture}
}
\caption{$P_5$ with center root vertex simplified is still $P_5$.}
\label{fig:p5top5}
\end{figure}%
\begin{figure}
\begin{center}
\subfigure[]{
\includegraphics[width=.375\textwidth]{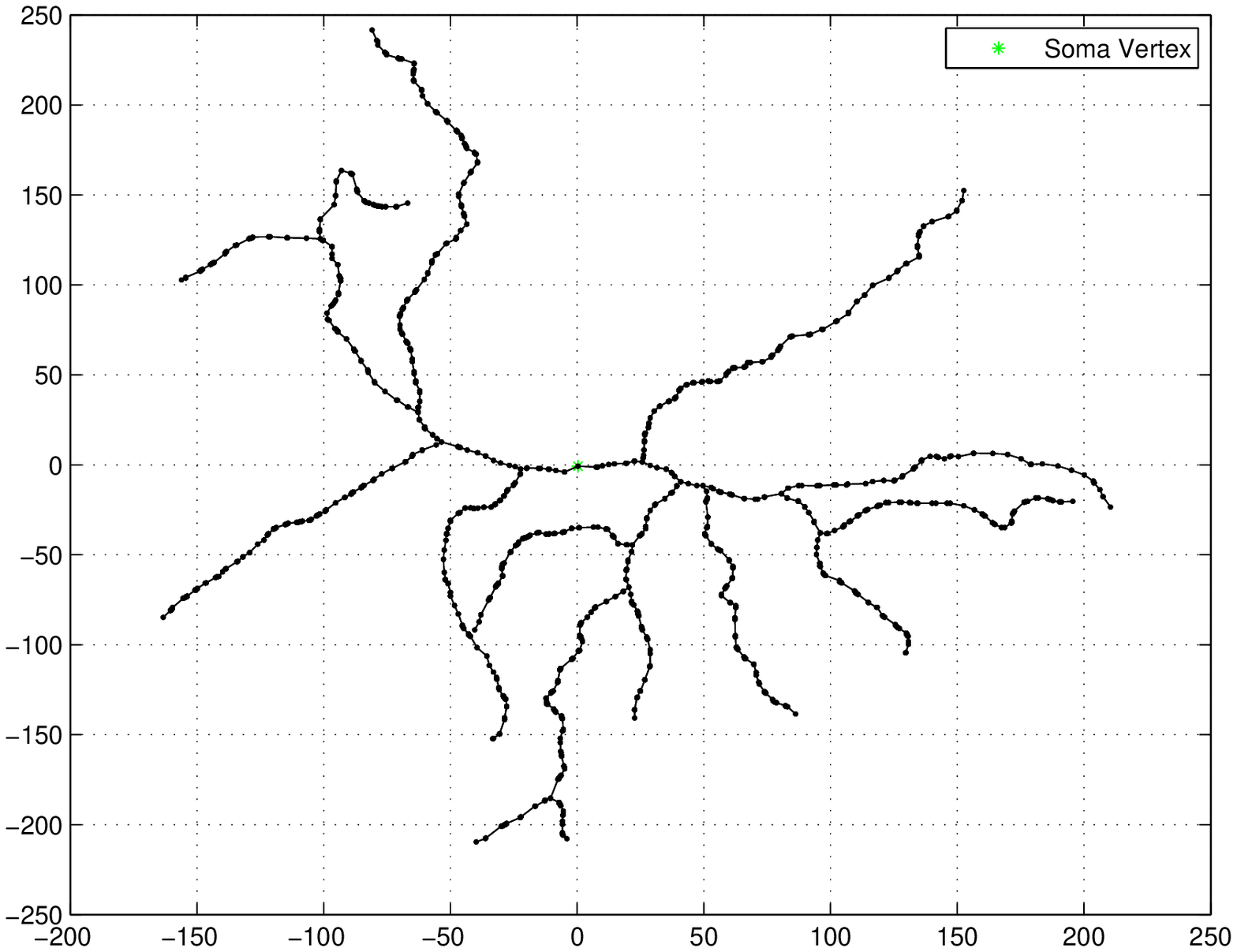}
\label{subfig:rgc100ori}
}
$\xrightarrow{\text{Tree Simplified}}$
\subfigure[]{
\includegraphics[width=.375\textwidth]{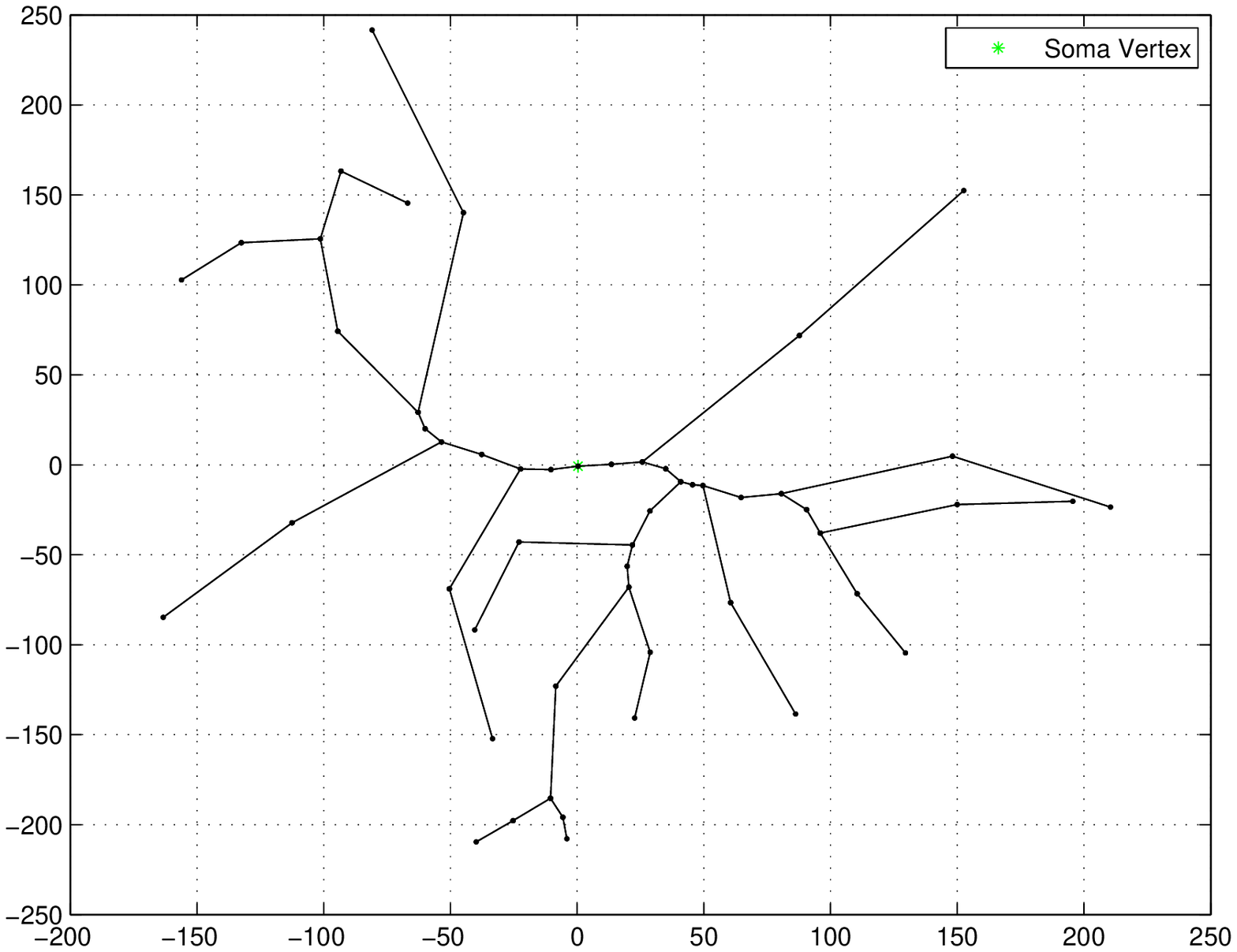}
\label{subfig:rgc100simp}
}
\\
\subfigure[]{
\includegraphics[width=.95\textwidth]{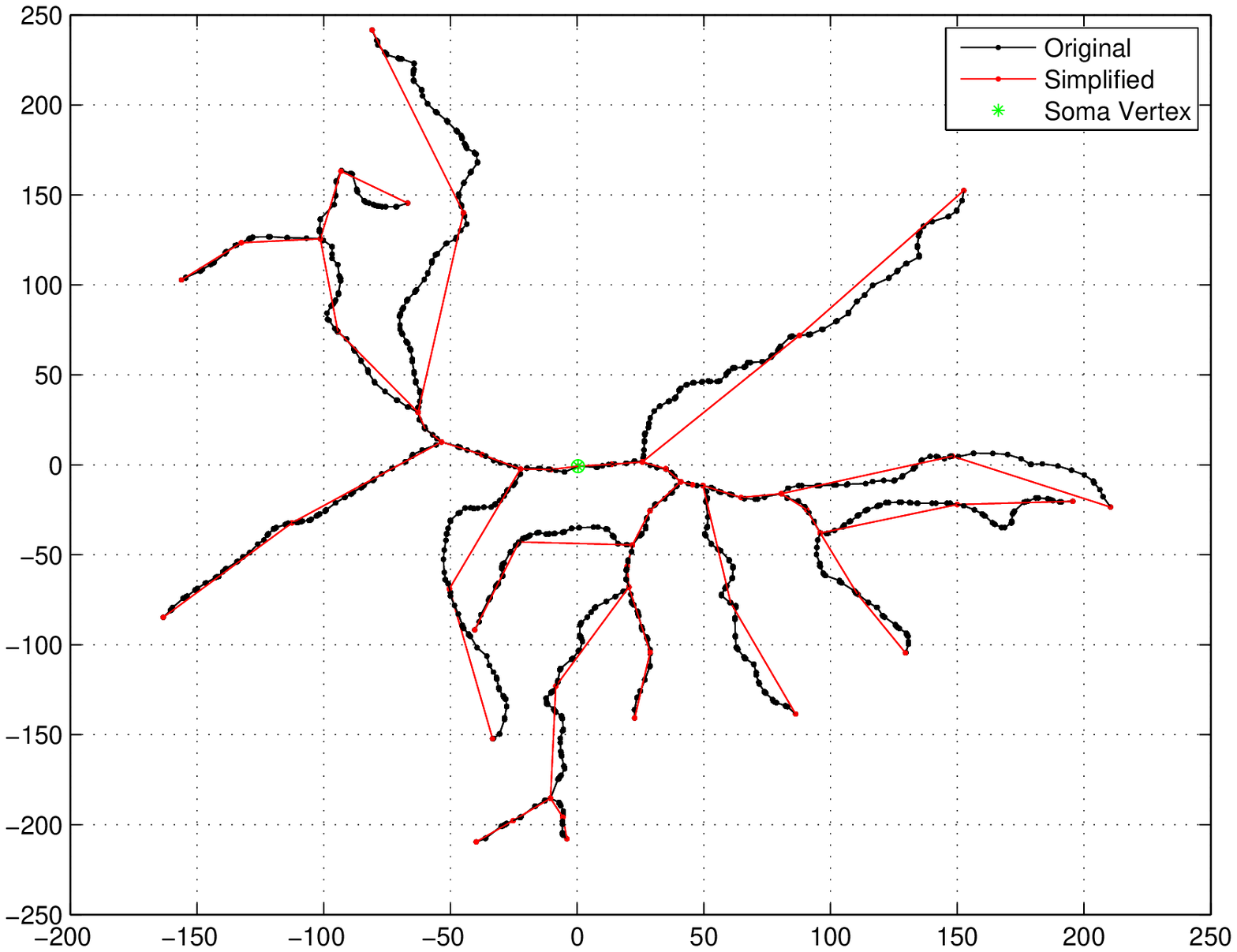}
\label{subfig:rgc100orisimpoverlay}
}
\end{center}
\caption{RGC \#100 from Cluster 6.  (a) Original, 1,154 vertices.  (b) Simplified, 53 vertices. (c) Overlayed.  Approximately 95\% vertices reduction from the original tree to the simplified tree.  Note that these plots are 2D projections of
the 3D dendritic structures since each original vertex has a 3D spatial coordinate.
The units of the horizontal and vertical axes are in $\mu \text{m}=10^{-6}$meter.}
\label{fig:rgc100orivssimp}
\end{figure}%
\begin{figure}
\begin{center}
\subfigure[]{
\includegraphics[height=.425\textheight]{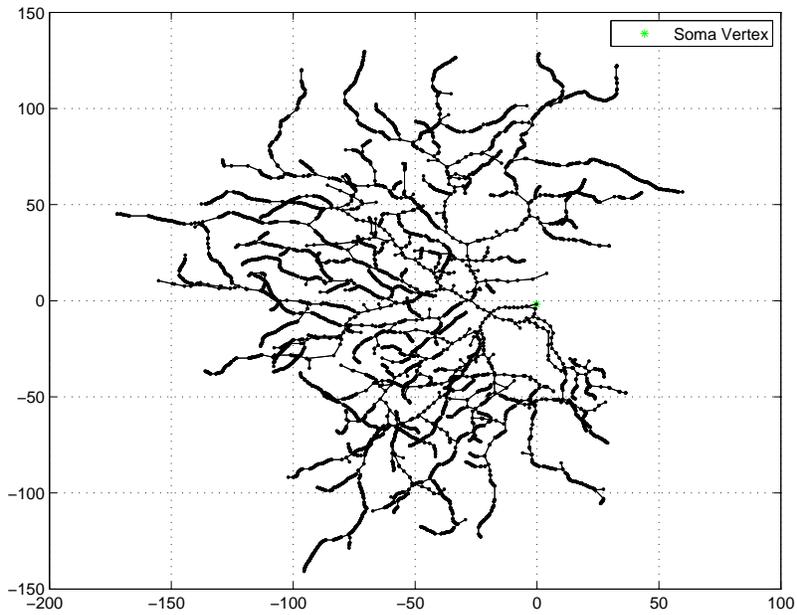}
\label{subfig:rgc60ori}
} \\
\subfigure[]{
\includegraphics[height=.425\textheight]{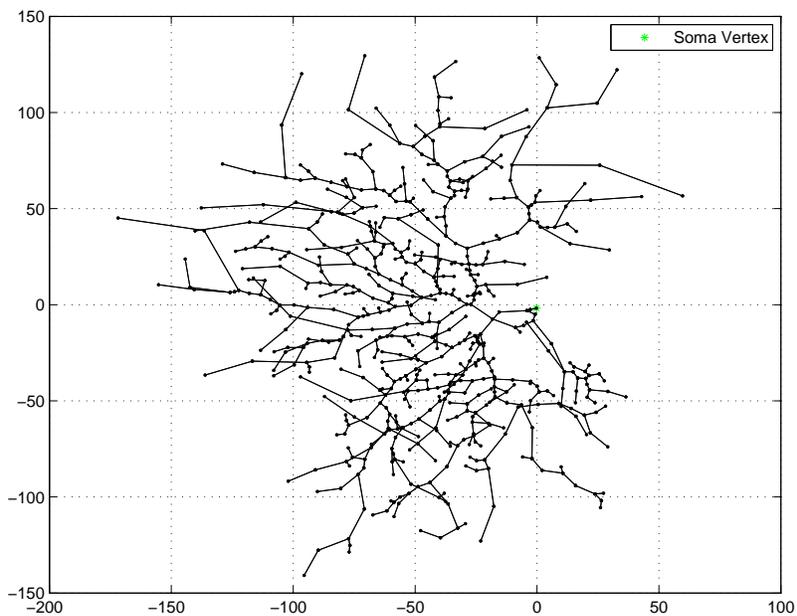}
\label{subfig:rgc60simp}
}
\end{center}
\caption{RGC \#60 from Cluster 1.  (a) Original, 5,636 vertices.  (b) Simplified, 612 vertices.  Approximately 89\% vertices reduction from the original tree to the simplified tree. Note that ``spines'' are preserved. }
\label{fig:rgc60orivssimp}
\end{figure}%
\begin{table}
\centering
\begin{tabular}{|*{8}{r|}}
\hline
\multicolumn{8}{|c|}{Vertex Count Statistical Information} \\
\hline
&&&&&&&\\
Cluster \# & \# RGCs & $\overline{\Lvert{V_o}}$ & $\sigma_o$ & $\overline{\Lvert{V_s}}$ & $\sigma_s$ & Red.~Avg. & Red.~Std. \\
\hline
1 &  9 & 4650.8 & 3693.68 & 866.8 & 189.98 & 70.23\% & 17.80\% \\
\hline
2 &  8 & 1562.4 & 565.18 & 337.1 & 167.65 & 78.46\% & 5.78\% \\
\hline
3 & 18 & 2262.4 & 1448.48 & 422.7 & 169.37 & 77.83\% & 8.27\% \\
\hline
4 &  8 & 9378.1 & 7369.27 & 645.0 & 300.54 & 88.38\% & 8.50\% \\
\hline
5 & 10 & 2778.1 & 1841.34 & 343.6 & 115.20 & 85.02\% & 6.68\% \\
\hline
6 &  9 & 758.7 & 172.16 & 39.6 & 19.60 & 94.55\% & 3.11\% \\
\hline
7 & 15 & 3245.0 & 3299.06 & 333.3 & 52.40 & 83.44\% & 8.46\% \\
\hline
8 & 19 & 3323.5 & 2125.14 & 302.7 & 49.42 & 87.25\% & 6.81\% \\
\hline
9 & 21 & 3113.9 & 2021.73 & 210.7 & 60.25 & 90.23\% & 6.21\% \\
\hline
10 & 13 & 3561.8 & 2727.46 & 173.6 & 35.98 & 91.22\% & 6.18\% \\
\hline
11 & 12 & 4668.8 & 2697.87 & 981.7 & 256.64 & 74.53\% & 10.95\% \\
\hline
12 &  9 & 8273.1 & 6809.86 & 599.3 & 189.34 & 87.02\% & 10.89\% \\
\hline
13 & 19 & 5043.2 & 2744.18 & 724.8 & 205.49 & 80.62\% & 11.56\% \\
\hline
14 &  8 & 3684.0 & 2763.62 & 446.9 & 144.29 & 83.49\% & 7.10\% \\
\hline
no &  1 & 5590.0 & 0.00 & 103.0 & 0.00 & 98.16\% & 0.00\% \\
\hline\hline
Total & 179 & 3852.5 & 3565.28 & 442.0 & 293.10 & 83.99\% & 10.89\% \\
\hline
\end{tabular}
\caption{Vertex count statistical information for the original and simplified 
dendritic tree for each cluster of RGCs.  These clusters were identified by
Coombs et al.~\citep{COOMBS-ALL-2006} from their morphological analysis 
followed by the hierarchical clustering technique.
Cluster `no' indicates a singular dendritic tree that does not belong to any of
the 14 distinct clusters.
The 3rd and 5th columns are the \emph{average} number of vertices across all 
trees for each cluster, for the original and simplified tree, respectively.  
The 4th and 6th columns represent the \emph{standard deviation} of the number of
vertices across all trees for each cluster.  The final two columns represent the
\emph{average percentage reduction} and \emph{reduction standard deviation} of 
number of vertices from the original to the simplified dendritic tree.
The subscripts `o' and `s' indicate `original' and `simplified', respectively.
}
\label{tab:orisimpstats}
\end{table}

\section{Eigenvalue `Plateaux' Phenomenon}
\label{sec:ewmult}
After simplifying all of our dendritic trees in our dataset, we computed their
Laplacian eigenvalues.  We then observed an ``eigenvalue plateaux'' phenomenon
for each of the newly formed simplified trees.  
In Figures~\ref{fig:rgc102simpewplot} and~\ref{fig:rgc108simpewplot}, 
we display a couple of simplified trees and their respective eigenvalue plots.  
The eigenvalues which form these plateaux are approximately $0.3820$ and 
$2.6180$.  
In fact, as we will explain shortly, these values are more precisely written as
$\lm \define \frac{3-\sqrt{5}}{2} = 2 - 2 \cos \frac{\pi}{5} \approx 0.3820$ and 
$\lp \define \frac{3+\sqrt{5}}{2} = 2 - 2 \cos \frac{3\pi}{5} \approx 2.6180$.  
It is interesting to note that the multiplicity of each of those eigenvalues 
are exactly the same.
In Table~\ref{tab:simpewmstats}, we present some statistical information on 
the eigenvalue multiplicity across all trees for each cluster.
For these dendritic trees, 
the multiplicities of $\lm$ were determined numerically
by counting all the eigenvalues lying within the interval of width of 
$2 \times 10^{-10}$ centered at $\lm$, and the same procedure was used for $\lp$.

\begin{figure}
\begin{center}
\subfigure[]{
\includegraphics[height=.425\textheight]{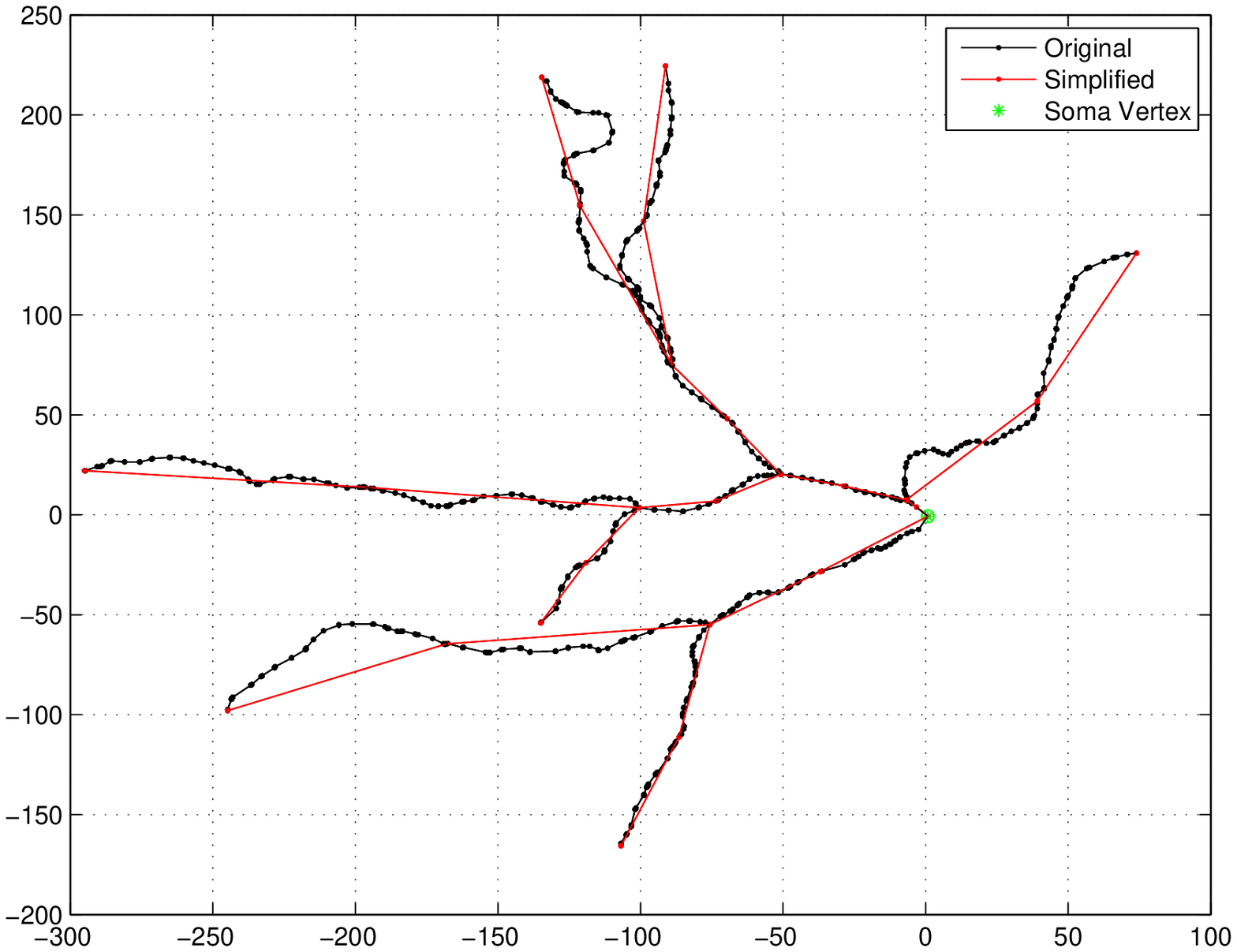}
\label{subfig:rgc102orisimpoverlay}
}
\\
\subfigure[]{
\includegraphics[height=.425\textheight]{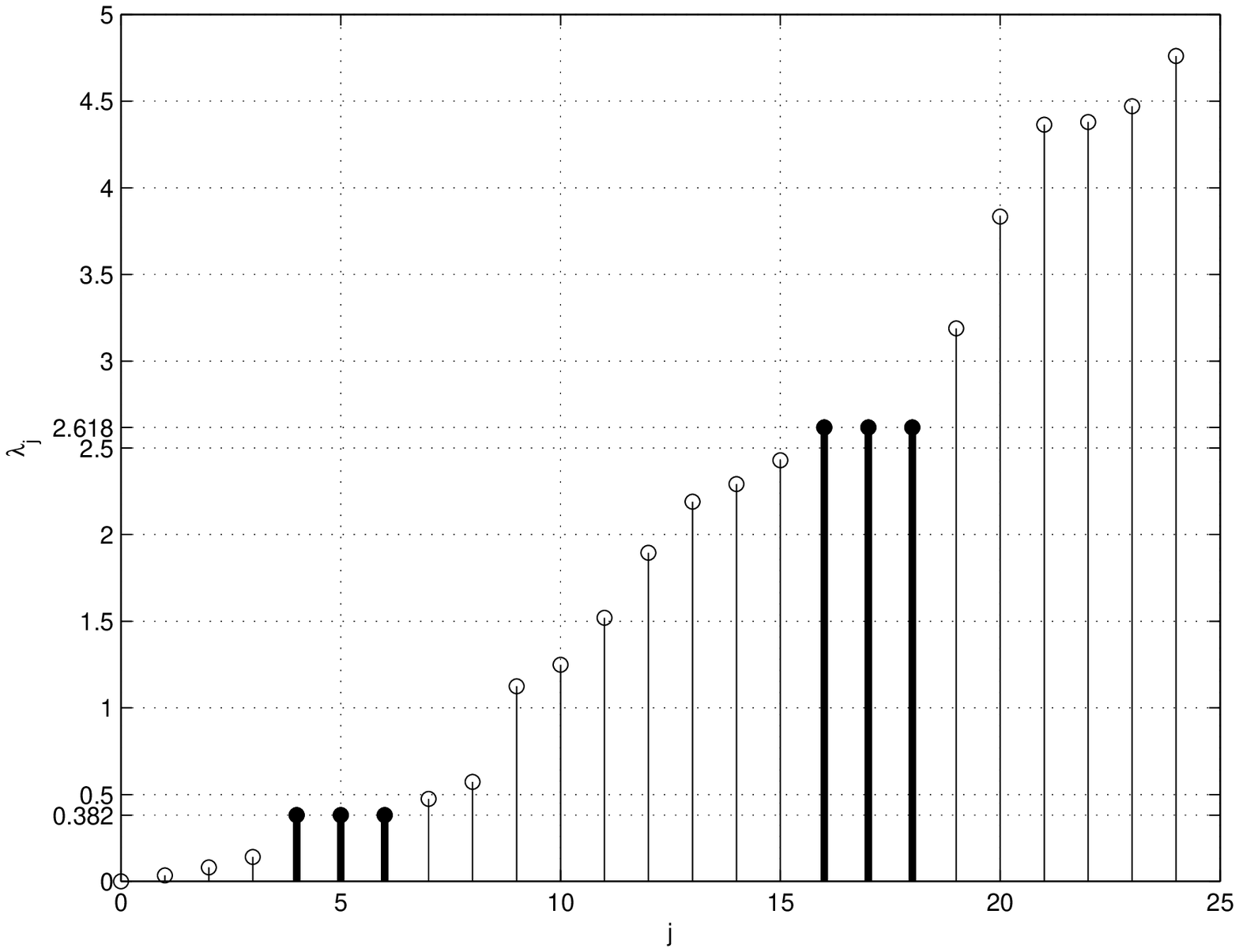}
\label{subfig:rgc102simpewplot}
}
\end{center}
\caption{RGC \#102 of Cluster 6.  (a) Simplified tree overlaid on original dendritic tree.  (b) Eigenvalue distribution of the simplified tree.  Note that $m_T(\lm)=m_T(\lp)=3$.}
\label{fig:rgc102simpewplot}
\end{figure}%
\begin{figure}
\begin{center}
\subfigure[]{
\includegraphics[height=.425\textheight]{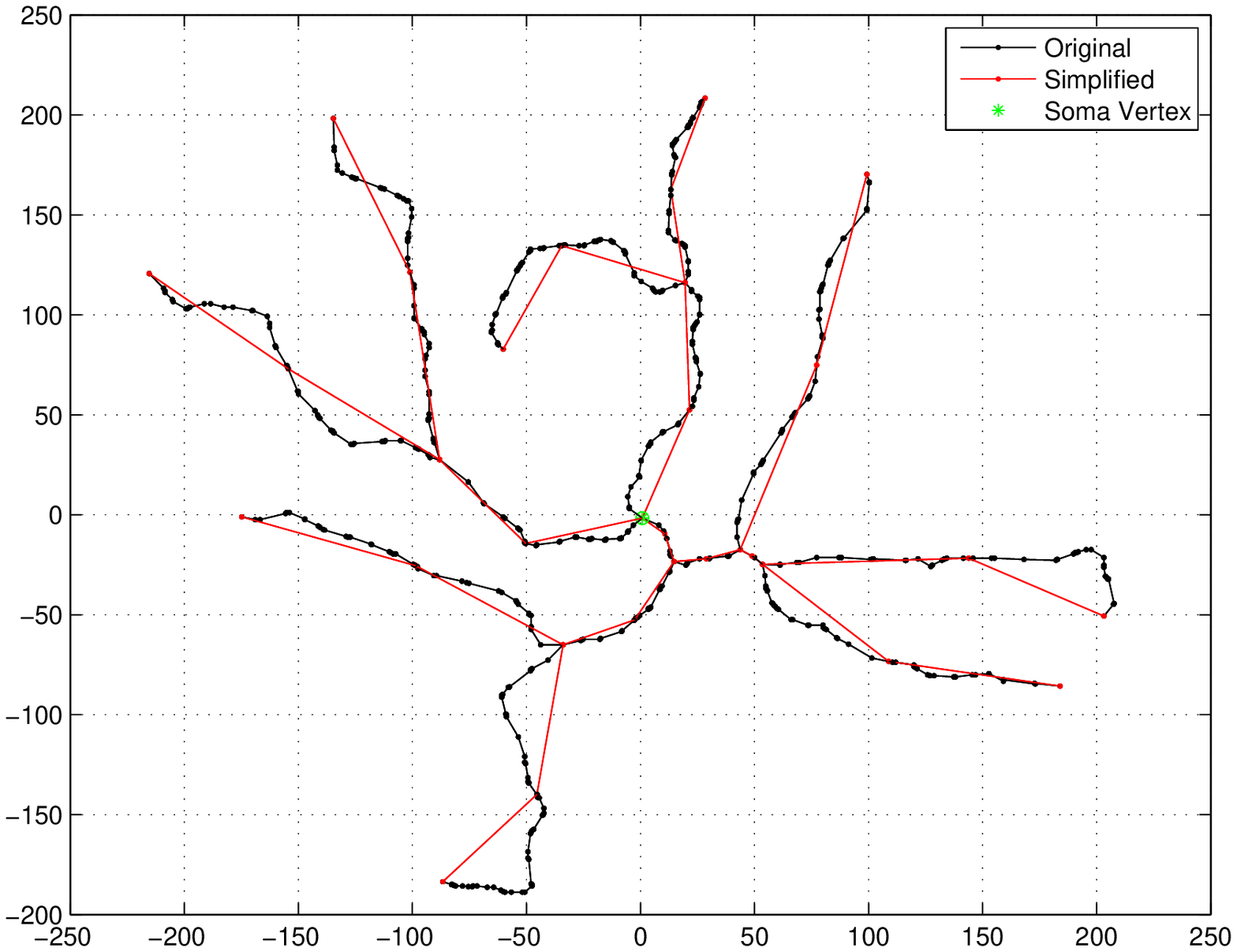}
\label{subfig:rgc108orisimpoverlay}
}
\\
\subfigure[]{
\includegraphics[height=.425\textheight]{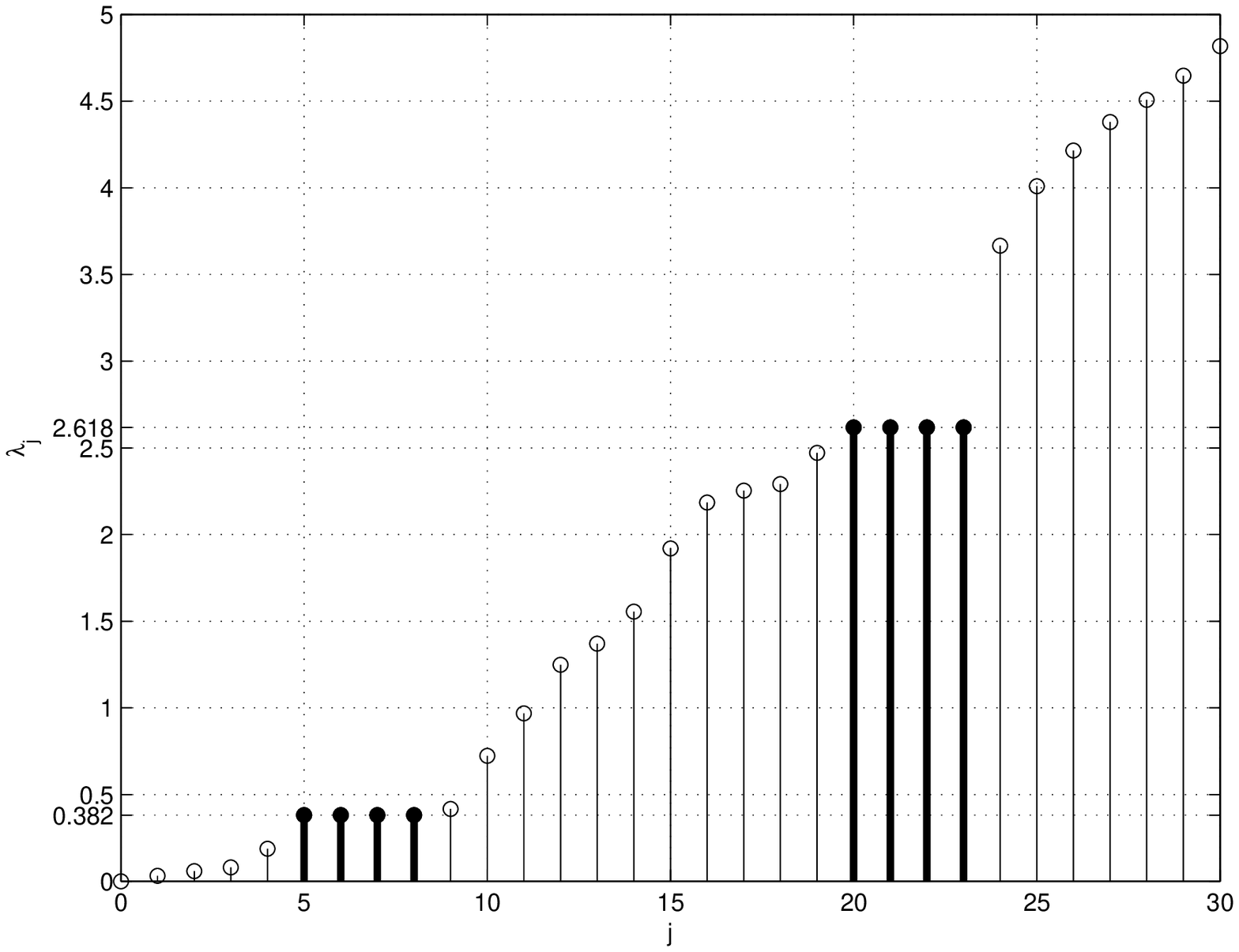}
\label{subfig:rgc108simpewplot}
}
\end{center}
\caption{RGC \#108 of Cluster 6.  (a) Simplified tree overlaid on original dendritic tree.  (b) Eigenvalue distribution of the simplified tree.  Note that $m_T(\lm)=m_T(\lp)=4$.}
\label{fig:rgc108simpewplot}
\end{figure}%
\begin{table}
\centering
\begin{tabular}{|*{6}{r|}}
\hline
\multicolumn{6}{|c|}{Eigenvalue Multiplicity Statistical Information} \\
\hline
&&&&& \\
Cluster \# & $\overline{m_T\Lpar{\Lbra{\lm,\lp}}}$ & $\sigma_m$ & Cluster \# & $\overline{m_T\Lpar{ \Lbra{\lm,\lp} }}$ & $\sigma_m$ \\
\hline
1 & 29.3 & 15.94 & 2 & 20.2 & 10.27 \\
\hline
3 & 34.2 & 11.45 & 4 & 58.8 & 19.54 \\
\hline
5 & 24.8 & 11.43 & 6 & 7.6 & 3.98 \\
\hline
7 & 36.3 & 10.19 & 8 & 31.2 & 8.47 \\
\hline
9 & 26.3 & 6.85 & 10 & 20.8 & 9.43 \\
\hline
11 & 68.7 & 39.04 & 12 & 36.7 & 10.99 \\
\hline
13 & 34.6 & 9.99 & 14 & 34.0 & 8.25 \\
\hline
no & 20.0 & 0.00 & All & 32.9 & 19.61 \\
\hline
\end{tabular}
\caption{Summary of the Laplacian eigenvalue multiplicity of the simplified dendritic trees of our dataset separated by each cluster.  The 2nd and 5th columns contain the \emph{average multiplicity} of the eigenvalues $\lm$ and $\lp$, while the 3rd and 6th columns contain the \emph{standard deviation} of these specific eigenvalue multiplicities. Note that $\overline{m_T\Lpar{\Lbra{\lm,\lp}}}
= 2 \overline{m_T\Lpar{\lm}} = 2 \overline{m_T\Lpar{\lp}}$.}
\label{tab:simpewmstats}
\end{table}

Below we mention a simple example of a tree which contains these eigenvalues and a type of tree that has identical multiplicities of $\lm$ and $\lp$.
\begin{Example}
The simplest tree possessing Laplacian eigenvalues $\lm$ and $\lp$ is $P_5$.
It is well known that the eigenvalues of $P_n$ are $2-2\cos{\frac{k\pi}{n}}$ 
for $k=0,1,\dots,n-1$; see, e.g., \citep{STRANG-1999}.
Hence for $n=5$, we have $\lm=2-2\cos{\frac{\pi}{5}}$ and
$\lp=2-2\cos{\frac{3\pi}{5}}$.
\end{Example}

Now the following proposition demonstrates a concrete example of the
existence of $\lm$ and $\lp$ with multiplicities.
\begin{Proposition} 
Let $T(V,E) = S(k.2) \define S(\underbrace{2,2,\dots,2}_k)$ be a starlike
tree with $k>1$ branches with each branch containing 2 vertices 
so that $\Lvert{V}=2k+1$.  Then $m_T(\lm)=m_T(\lp)=k-1$.
\end{Proposition}
We note that $P_5 = S(2.2)$.
\begin{proof}
This easily follows from the following lemma due to Das:
\begin{Lemma}[Das {\citep[Lemma 3.1]{DAS-2007}}]
\label{lem:das}
The Laplacian eigenvalues of starlike tree $S(k.m)$ are 
\begin{equation}
\label{eqn:das}
2+2\cos\left(\frac{p \pi}{2m+1}\right), \quad p=2, 4, \ldots, 2m, 
\end{equation}
and each of multiplicity $k-1$.  Also the remaining eigenvalues satisfy
the following system of equations:
$$
\left\lbrace
\begin{aligned}
\lambda x_1 &= x_1 - x_2 \\
\lambda x_i &= 2x_i-x_{i-1}-x_{i+1} \quad i=2, 3, \ldots, m \\
\lambda x_{m+1} &= kx_{m+1} - kx_m 
\end{aligned} \right.
$$
\end{Lemma}
Hence, setting $m=2$ in this lemma, the eigenvalues \eqref{eqn:das} are of 
the form  $2 + 2 \cos \frac{p \pi}{5}$, $p=2, 4$.
But $2+2\cos\frac{2\pi}{5}=2-2\cos\frac{3\pi}{5}=\lp$, and
$2+2\cos\frac{4\pi}{5}=2-2\cos\frac{\pi}{5}=\lm$.
Their multiplicies are $k-1$ as Das's lemma guarantees.
\end{proof}

We now explain why our simplified dendritic trees have the eigenvalue plateaux
phenomenon that we had observed.  First, let us define some notation that will
be used in our theorem below.  Let $G(V,E)$ be a simple, connected, undirected,
and unweighted graph.  Let
\eqn{
\vo \define \Lbra{v \in V \cond d_v = 1} \subset V
\label{eqn:vo}
}
be the set of \emph{pendant vertices} and 
\eqn{
\vto \define \Lbra{v \in V \cond d_v=2~\text{and}~\exists u \in \vo~\text{s.t.}~u \sim v}
\label{eqn:vto}
}
be the set of \emph{pendant neighbors of degree 2}, and
\eqn{
\vi \define \Lbra{v \in V \cond d_v \geq 3~\text{and}~\exists u \in \vto~\text{s.t.}~u \sim v} 
\label{eqn:vi}
}
be the set of \emph{vertices of degree 3 or greater which are adjacent to 
vertices in $\vto$}.  For every $v \in \vi$, let us define the following two 
quantities
\eqn{
c(v) \define \Lvert{\Lbra{v' \in \vto \cond v' \sim v}}
\label{eqn:cv}
}
and
\eqn{
\tvi \define \sum_{v \in \vi} \Lpar{c(v)-1}.
\label{eqn:tvi}
}
Note that there may be vertices of degree 3 or greater in $V$ that do not 
belong to $\vi$.  
The following theorem explains the eigenvalue plateaux phenomenon 
\emph{not only for simplified trees but also for more general graphs}.
\begin{Theorem}
\label{thm:ewmult}
Let $G(V,E)$ be a simple, connected, undirected, and unweighted graph with 
$n=\Lvert{V}$. 
Let $\lm$ and $\lp$ be as defined previously in this section.  
Suppose $\tvi \geq 1$, then
\eqn{
m_G(\lm) = m_G(\lp) \geq \tvi .
\label{eqn:multlmpineq}
}
In other words, the multiplicity of the graph Laplacian eigenvalues $\lm$ 
and that of $\lp$ are the same and at least $\tvi$.
\end{Theorem}
\begin{proof}
Let $\kappa \define \Lvert{\vi}$.  If $\kappa=0$, then obviously 
$\tvi = 0 \leq m_G(\lm)$.  Similarly, for $m_G(\lp)$.  Suppose $\kappa>0$, 
then $\vi \neq \emptyset$.  For every $v \in \vi$, we have $c(v)\geq1$ by 
the definitions of \eqref{eqn:cv} and \eqref{eqn:vi}.

Suppose $\vi = \Lbra{v_1,\dots,v_\kappa}$.
A Laplacian matrix of $G$ takes the form
\eqn{
L(G)=
\left[
\begin{array}{ccc|ccc|ccc}
B_1   &        &         & \br_1 &        &           &  &  & \\
      & \ddots &         &       & \ddots &           &  &  & \\
      &        & B_\kappa &       &        &  \br_\kappa &  &  & \\ \hline

\br_1^\transp &  &        &       &        &            &  &  & \\       
      & \ddots &         &       &  C_1   &          &  & C_2 & \\
      &      & \br_\kappa^\transp & &       &          &   &  & \\ \hline

      &      &           &       &       &            &  &    & \\
      &      &           &       &  C_3  &            &  & C_4 & \\
      &      &           &       &       &            &  &     &
\end{array}
\right],
\label{eqn:Lorg}
}
where for each $j$ with $1 \leq j \leq \kappa$, 
and $B_j \define \diag(Q, \ldots, Q) \in \Rf^{2c(v_j) \times 2c(v_j)}$, i.e.,
a block diagonal matrix with $c(v_j)$ copies of 
$Q \define \begin{bmatrix} 2 & -1 \\ -1 & 1 \end{bmatrix}$, 
and $\br_j \in \Rf^{2c(v_j) \times 1}$ is a vector 
of $c(v_j)$ stacks of the vector $\br \define \begin{bmatrix} -1 \\ 0 \end{bmatrix}$, i.e.,
$$
\br_j \define
\begin{bmatrix}
\br \\
\br \\
\vdots \\
\br
\end{bmatrix} \in \Rf^{2c(v_j) \times 1}.
$$
Note that the eigenvalues of $Q$ is $\lm$ and $\lp$.

To describe the block matrices $C_\ell$ for $\ell=1,2,3,4$, first let us define
\eqn{
\vtoi \define \Lbra{v \in \vto \cond \exists v' \in \vi~\text{s.t.}~v \sim v'} \subseteq \vto
\label{eqn:vtoi}
}
and
\eqn{
\voi \define \Lbra{v \in \vo \cond \exists v' \in \vtoi~\text{s.t.}~v \sim v'} \subseteq \vo.
\label{eqn:voi}
}
Note that $\left| \vtoi \right| = \left| \voi \right| = \sum_{j=1}^\kappa c(v_j)$.
Let $V_r \define V \setminus \Lpar{\vi \cup \vtoi \cup \voi}$.  
$C_1$ is a $\kappa \times \kappa$ matrix whose diagonal entries are 
$\dv{1},\dv{2},\dots,\dv{\kappa}$ and whose off-diagonal entries
depends on the interactions (edges) between vertices within $\vi$.
The block matrices $C_2$ and $C_3$ correspond to the interactions between 
vertices in $\vi$ and $V_r$, while the block matrix $C_4$ corresponds to the
interactions between vertices within $V_r$.

We need only consider the case when $c(v_j)>1$, since if $c(v_j)=1$, then it 
does not contribute to the sum in \eqref{eqn:tvi}.  Our strategy is to construct 
a set of $c(v_j)-1$ eigenvector(s) associated to the eigenvalue $\ew$, where 
$\ew=\lm$ or $\lp$ or equivalently when $\ew$ satisfies the characteristic 
equation $\ew^2-3\ew+1=0$.  First, let 
\eqn{
\bxew \define
\begin{bmatrix}
\frac{-1}{2-\ew} \\
-1
\end{bmatrix}
\label{eqn:bx}
}
and let $\zv_k$ denote the zero vector of length $k$.

For $\ell=1,\dots,c(v_j)-1$, let $\by_\ell$ be the vector of length $2c(v_j)$ 
where $\bxew$ occupies the 1st and 2nd entries of $\by_\ell$ and $-\bxew$ occupies the $2\ell+1$ and $2\ell+2$ entries of $\by_\ell$, i.e.,
$$
\by_\ell=
\begin{bmatrix}
\bxew \\
\zv_2 \\
\vdots \\
\zv_2 \\
-\bxew \\
\zv_2 \\
\vdots \\
\zv_2
\end{bmatrix}
\begin{array}{l}
\\ \\ \\ \leftarrow 2\ell+1~\text{and}~2\ell+2~\text{positions.} \\ \\ \\
\end{array}
$$
Therefore,
$$
\by_\ell \in 
\underbrace{
\Lbra{
\begin{bmatrix}
\bxew \\ -\bxew \\ \zv_2 \\ \vdots \\ \vdots \\ \vdots \\ \zv_2
\end{bmatrix},
\begin{bmatrix}
\bxew \\ \zv_2 \\ -\bxew \\ \zv_2 \\ \vdots \\ \vdots \\ \zv_2
\end{bmatrix},
\dots,
\begin{bmatrix}
\bxew \\ \zv_2 \\ \vdots \\ \vdots \\ \zv_2 \\ -\bxew \\ \zv_2
\end{bmatrix},
\begin{bmatrix}
\bxew \\ \zv_2 \\ \vdots \\ \vdots \\ \vdots \\ \zv_2 \\ -\bxew
\end{bmatrix}
}
}_{c(v_j)-1~\text{vector(s)}}.
$$
Note that,
\eqn{
Q\bxew=
\begin{bmatrix} 2 & -1 \\ -1 & 1 \end{bmatrix}
\begin{bmatrix} \frac{-1}{2-\ew} \\ -1 \end{bmatrix}
=
\begin{bmatrix} \frac{-\ew}{2-\ew} \\ \frac{\ew-1}{2-\ew} \end{bmatrix}
=
\ew
\begin{bmatrix} \frac{-1}{2-\ew} \\ -1 \end{bmatrix}
=
\ew \bxew,
\label{eqn:qewev}
}
where the second equality from the right of \eqref{eqn:qewev} holds if
$\ew=\lm$ or $\lp$. 
From this it is easy to see that
$$
B_j \by_\ell = \ew \by_\ell,\quad \ell=1,\dots,c(v_j)-1.
$$
Hence, $\by_\ell$ is an eigenvector of $B_j$ with eigenvalue $\ew$ for each $\ell=1,\dots,c(v_j)-1$.  


It is clear that $\br^\transp_j \by_\ell = 0$ for all $\ell=1,\dots,c(v_j)-1$
using the definition of both vectors.
Hence we can construct our eigenvectors for $L(G)$ corresponding to
the eigenvalue $\lambda = \lm$ or $\lp$ as
\eqn{
\ev{\Lpar{\ell,j}}=
\begin{bmatrix}
\zvcv{1} \\
\vdots \\
\zvcv{j-1} \\
\by_\ell \\
\zvcv{j+1} \\
\vdots \\
\zvcv{\kappa} \\
\zv_{n-\sum_{i=1}^\kappa 2c(v_i)}
\end{bmatrix},
\label{eqn:ev-ell-j}
}
for each $j=1,\dots,\kappa$ and $\ell=1,\dots,c(v_j)-1$. 
It is clear that the muliplicity of $\lm$ accounted by the vectors of the form
\eqref{eqn:ev-ell-j} is $\sum_{j=1}^\kappa (c(v_j)-1)$, which is exactly equal
to $\tvi$.  The same argument also holds for $\lp$.  
Therefore the inequality of \eqref{eqn:multlmpineq} is achieved.  

Finally, in order to show $m_G(\lm)=m_G(\lp)$, let us first recall that
the characteristic polynomial of $L(G)$ is a monic polynomial of
integer-valued coefficients.  
Moreover, the numbers $\lm$ and $\lp$ are the so-called
\emph{algebraic integers} in the field of $\Q\left[\sqrt{5}\right]$ \citep[Chap.~13]{ARTIN}.
Combining these with the fact that the eigenvalues of $L(G)$ are
nonnegative real numbers, thanks to the Galois theory \citep[Chap.~16]{ARTIN},
we know that if $\lm$ is a root of such characteristic polynomial, 
then its ``real conjugate'' $\lp$ must also be a root and their multiplicities
must be the same, i.e., $m_G(\lm)=m_G(\lp)$.
\end{proof}
We note that 
the eigenvectors constructed in \eqref{eqn:ev-ell-j} do not include the remaining
eigenvectors of $L$, especially those related to the $C_1,C_2,C_3,$ and $C_4$ 
portions.

We now present some examples below to demonstrate our Theorem~\ref{thm:ewmult}.
\begin{Example}[$m_G(\lm)=m_G(\lp)=\tvi$]
Figure~\ref{fig:rgc102simpthmex} shows that $\tvi=3$.  We can see this by 
observing the bordered rectangles encompassing the indicated vertices.  
Here, we have $v_i \in \vi$ with $c(v_i)=2$ for each $i=1,2,3$ in 
Figure~\ref{fig:rgc102simpthmex}.  By comparison to the results displayed in
Figure~\ref{subfig:rgc102simpewplot}, we see that $m_G(\lm)=m_G(\lp) = 3 = \tvi$.
\end{Example}
\begin{figure}
\centering
\includegraphics[width=\textwidth]{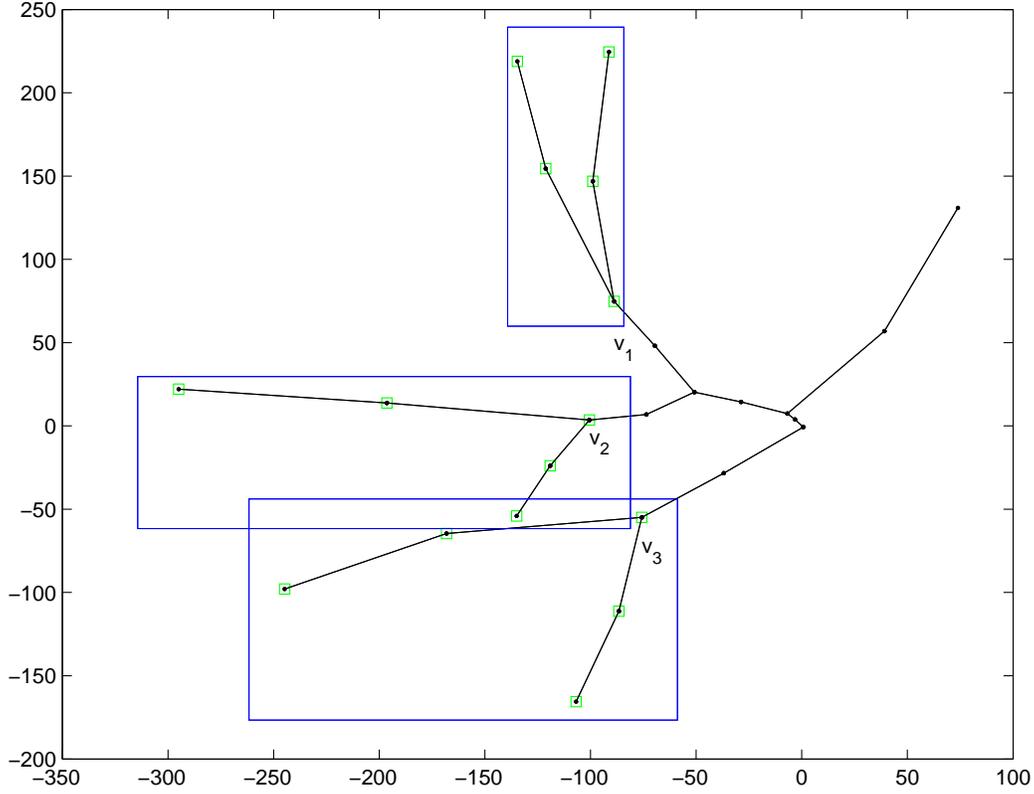}
\caption{RGC \#102 of Cluster 6 dendritic tree simplified.  $v_i \in \vi$ for each $i=1,2,3$ with $c(v_i)=2$ for $i=1,2,3$, hence $\tvi=3$.}
\label{fig:rgc102simpthmex}
\end{figure}

\begin{Example}[$m_G(\lm)=m_G(\lp) \gneqq \tvi$]
Figure~\ref{fig:rgc96simpthmex} shows an example where 
$13=m_G(\lm) = m_G(\lp) \gneqq \tvi=12$.  The rectangles in this figure 
indicates locations of vertices $v_i \in \vi$ for $i=1,2,\dots,12$, which 
contribute to $\tvi$. 
\end{Example}
\begin{figure}
\centering
\includegraphics[width=\textwidth]{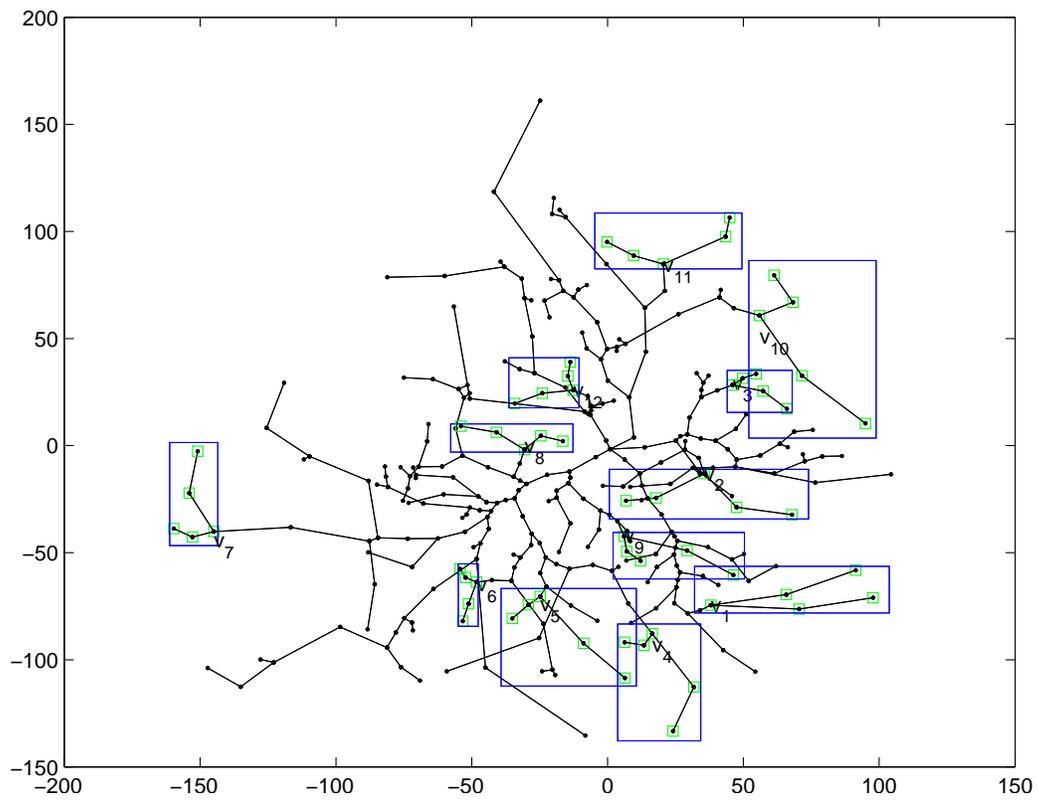}
\caption{RGC \#96 of Cluster 3 dendritic tree simplified.  $\tvi=12$, yet 
$m_G(\lm)=m_G(\lp)=13$.  Vertices of interest are surrounded by green outlined 
squares and encompassed by a rectangle.}
\label{fig:rgc96simpthmex}
\end{figure}

\section{Discussion}
\label{sec:conc}
In this article, we introduced a tree simplification procedure that yielded the
highly pronounced ``eigenvalue plateaux'' phenomenon in all of our simplified
dendritic trees.  We explained the reason of the occurrence of this phenomenon
by splitting the vertex set $V$ of a graph (more general than such
simplified trees) into a set of mutually exclusive subsets, 
$\vi$, $\vtoi$, $\voi$, and $V_r$ followed by the explicit construction of
eigenvectors corresponding to the multiple eigenvalues.

We now discuss a potential generalization to our Theorem~\ref{thm:ewmult}.
Recall the following theorem on the relationship between the 
pendant vertices and pendant neighbors:
\begin{Theorem}[Faria \citep{FARIA-1985}; see also Merris \citep{MERRIS-1994}]
Let $G$ be a graph, and let $p(G)$ and $q(G)$ be the number of pendant vertices
and pendant neighbors in $G$, respectively.  Then,
$$
p(G)-q(G) \leq m_G(1).
$$
\end{Theorem}
This inequality was used to derive a spectral feature for clustering dendritic
trees in \citep{SAITO-ALL-2009}.
In this article, if we refer to 
the vertices in $\vi$ as \emph{pendant $P_2$ neighbors} and 
the vertices in the set $\vtoi$ as \emph{pendant $P_2$ vertices},
with $q_2(G)$ and $p_2(G)$ as their cardinalities, respectively,
then it is clear that $p_2(G)=\left|\vtoi\right|=\sum_{v\in\vi}c(v)$, $q_2(G)=\left|\vi\right|=\kappa$.
Hence, Theorem~\ref{thm:ewmult} can be rewritten as
$$
p_2(G)-q_2(G) \leq m_G(\lm) = m_G(\lp).
$$
More generally, let us define the notion of a \emph{pendant $P_j$ vertex} and a
\emph{pendant $P_j$ neighbor}.
A vertex $v$ in a graph $G(V, E)$ with $|V|=n$ is said to be a pendant $P_j$
vertex ($1 \leq j \leq n$) if the following conditions are satisfied:
1) $v$ is a trivial vertex; 2) $v$ is adjacent to a nontrivial vertex $u$ of 
degree greater than two; and 3) there is a pendant vertex $w \in V$ such that
$u$ and $w$ form a neighboring pair of nontrivial vertices with a path of
length $j$.  A vertex $u \in G(V,E)$ is said to be a pendant $P_j$ neighbor
if $u$ is a nontrivial vertex with degree greater than two and is adjacent 
to a pendant $P_j$ vertex.
We then have the following:
\begin{Conjecture}
Let $G(V,E)$ be a simple, connected, undirected, and unweighted graph with $|V|=n$.
Let $1 \leq j \leq n$, and let $p_j(G)$, $q_j(G)$ be the number of pendant 
$P_j$ vertices and the number of pendant $P_j$ neighbors, respectively.
Then,  
$$
p_j(G)-q_j(G) \leq m_G(\ew_s(G)),
$$
holds for some $s \in \{0, \ldots, n-1\}$.
\end{Conjecture}
This conjecture is certainly true if $G(V,E)=S(k.m)$ where $1 \leq j \leq m$.
This is because we can easily show from the definitions:
$$
p_j(G) = \begin{cases}
0 & \text{if $1 \leq j < m$;} \\
k & \text{if $j = m$}
\end{cases}
\quad \text{and} \quad
q_j(G) = \begin{cases}
0 & \text{if $1 \leq j < m$;} \\
1 & \text{if $j = m$}
\end{cases}
$$
imply
$$
p_j(G)-q_j(G) = \begin{cases}
0 & \text{if $1 \leq j < m$;} \\
k-1 & \text{if $j = m$}.
\end{cases}
$$
Since $m_G(\lambda_s(G)) \geq 0$ for every $s \in \{0, 1, \ldots, n-1\}$,
the former case of $p_j(G)-q_j(G) = 0 \leq m_G(\lambda_s(G))$ 
is certainly true.  If $j=m$, then there are $m$ eigenvalues 
with multiplicity $k-1$ as Lemma~3.1 of Das \citep{DAS-2007} shows.
Hence, $p_j(G)-q_j(G) = k-1 \leq m_G(\lambda_s(G))$ also holds for certain
$s \in \{0, 1, \ldots, n-1\}$.


\begin{hide}
Finally, according to our collaborator Yuji Nakatsukasa, there has been work on simultaneously occurring eigenvalues, e.g., in our case $\lm$ and $\lp$ occured simultaneously with equal multiplicities in our experiments.  We expect that this generalization can be explained and could have some future research potential.
\end{hide}

\section*{Acknowledgments}
N.\ S.\ would like to thank Prof.\ Monica Vazirani of UC Davis for explaining
the basics of the roots of a monic polynomial with integer coefficients.
This research was partially supported by the following grants from the Office
of Naval Research: N00014-09-1-0041; N00014-09-1-0318; N00014-12-1-0177
as well as the National Science Foundation: DMS-0636297; DMS-1418799.